\pdfoutput=1

\def\editmode{0}
\def\reportmode{0}

\def\bibfilenames{WISENET}

\if\reportmode1
\documentclass[10pt,final,onecolumn]{IEEEtran}
\else
\documentclass[journal]{IEEEtran}
\fi
\usepackage{dsfont,enumerate}
\if\editmode1  
\usepackage{mathrsfs}
\usepackage{bbm}
\usepackage{amssymb}
\usepackage[english]{babel}
\usepackage[utf8]{inputenc}
\usepackage{algorithm}
\usepackage[noend]{algpseudocode}
\floatname{algorithm}{Procedure}
\usepackage[backend=bibtex,style=alphabetic,sorting=debug]{biblatex}
\DeclareFieldFormat{labelalpha}{\thefield{entrykey}}
\DeclareFieldFormat{extraalpha}{}
\bibliography{\bibfilenames}
\newcommand{\cmt}[1]{\noindent\textcolor{lightgreen}{\underline{[#1]}}} 
\newenvironment{myitemize}{\begin{itemize}}{\end{itemize}}
\newcommand{\myitem}{\item}

\newtheorem{theorem}{Theorem}
\else
\usepackage{mathrsfs}
\usepackage{graphicx}
\usepackage{amssymb}
\usepackage[english]{babel}
\usepackage[utf8]{inputenc}
\usepackage{algorithm}
\usepackage[noend]{algpseudocode}
\floatname{algorithm}{Algorithm}
\usepackage{cite}
\newcommand{\cmt}[1]{} 
\newenvironment{myitemize}{}{}
\newcommand{\myitem}{}

\newtheorem{theorem}{Theorem}
\fi 

\usepackage{fixltx2e}

\usepackage{graphicx}

\usepackage{subcaption}              

\usepackage[utf8]{inputenc}

\usepackage{amsfonts}
\usepackage{amsmath}
\usepackage{mathtools}

\usepackage{amssymb}

\usepackage{bm}

\usepackage{color,verbatim}
\usepackage{multirow}
\usepackage{accents}

\usepackage{theoremref}

\usepackage{url}

\newtheorem{myauxproblem}{Problem}

\newcounter{rulecounter}
\newcommand{\resetrule}{ \setcounter{rulecounter}{0}}
\resetrule

\newsavebox{\selvestebox}
\newenvironment{colbox}[1]
  {\newcommand\colboxcolor{#1}%
   \begin{lrbox}{\selvestebox}%
   \begin{minipage}{\dimexpr\columnwidth-2\fboxsep\relax}}
  {\end{minipage}\end{lrbox}%
   \begin{center}
   \colorbox{\colboxcolor}{\usebox{\selvestebox}}
   \end{center}}

\definecolor{orange}{rgb}{1,0.8,0}
\definecolor{gray}{rgb}{.9,0.9,0.9}
\definecolor{darkgray}{rgb}{.3,0.3,0.3}
\definecolor{darkblue}{rgb}{.1,0.0,0.3}
\definecolor{lightblue}{rgb}{0.7,0.7,1}
\definecolor{lightred}{rgb}{1,0.7,.7}
\definecolor{purple}{RGB}{204,153,255}
\definecolor{lightgray}{rgb}{.95,0.95,0.95}
\definecolor{lightgreen}{rgb}{0.6,0.8,0.6}
\definecolor{darkgreen}{rgb}{0.05,0.3,0.05}
\definecolor{pistachio}{RGB}{204,255,153}
\definecolor{paleturquoise}{RGB}{175,238,238}
\definecolor{yellow}{RGB}{255,255,153}

\newcommand{\hbm}[1]{{\hat{\bm #1}}}

 \newcommand{\define}{\triangleq}

\newtheorem{myproposition}{Proposition}
\newtheorem{myquestion}{Question}
\newtheorem{myquiz}{Quiz}
\newtheorem{myremark}{Remark}
\newtheorem{myproblemstatement}{Problem Statement}
\newtheorem{mylemma}{Lemma}
\newtheorem{mytheorem}{Theorem}
\newtheorem{mydefinition}{Definition}
\newtheorem{mycorollary}{Corollary}
\newtheorem{myexample}{Example}

\begin{document}

	\title{Dynamic Regret Analysis for Online Tracking of Time-varying Structural Equation Model Topologies}
	\if\reportmode1
	\author{Bakht Zaman,~\IEEEmembership{Student Member,~IEEE,} Luis Miguel Lopez Ramos,~\IEEEmembership{Member,~IEEE,}  \\ 
		and Baltasar Beferull-Lozano,~\IEEEmembership{Senior Member,~IEEE}
		\thanks{The work in this paper was supported by the SFI Offshore Mechatronics grant 237896/E30, the PETROMAKS Smart-Rig grant 244205 and the IKTPLUSS Indurb grant 270730/O70 from the Research Council of Norway.
		}
		\thanks{The authors  are with the WISENET Lab, Dept. of ICT,
			University of Agder, Jon Lilletunsvei 3, Grimstad, 4879 Norway. E-mails:\{bakht.zaman, luismiguel.lopez, daniel.romero, baltasar.beferull\}@uia.no.}
		\thanks{The material in this work was presented, in part, at CAMSAP 2017 \cite{zaman2017onlinetopology}.}
	}%
	\else
	\author{Bakht Zaman,~\IEEEmembership{Student Member,~IEEE,} Luis Miguel Lopez Ramos,~\IEEEmembership{Member,~IEEE,}\\ 
		and Baltasar Beferull-Lozano,~\IEEEmembership{Senior Member,~IEEE}\\
		WISENET Center,
		Department of ICT, University of Agder, Grimstad, Norway
		\thanks{The work in this paper was supported by the SFI Offshore Mechatronics grant 237896/E30, the PETROMAKS Smart-Rig grant 244205 and the IKTPLUSS Indurb grant 270730/O70 from the Research Council of Norway. }
		\thanks{The authors  are with the WISENET Center, Dept. of ICT,
			University of Agder, Jon Lilletunsvei 3, Grimstad, 4879 Norway. E-mails:\{bakht.zaman, luismiguel.lopez,  baltasar.beferull\}@uia.no.}
	}%
	\fi

	\maketitle

	\begin{abstract}
		 Identifying dependencies among variables in a complex system is an important problem in network science. Structural equation models (SEM) have been used widely in many fields for topology inference, because they are tractable and incorporate exogenous influences in the model. Topology identification based on static SEM is useful in stationary environments; however, in many applications a time-varying underlying topology is sought. This paper presents an online algorithm to track sparse time-varying topologies in dynamic environments and most importantly, performs a detailed analysis on the performance guarantees. The tracking capability is characterized in terms of a bound on the dynamic regret of the proposed algorithm. Numerical tests show that the proposed algorithm can track changes under different models of time-varying topologies. 
	\end{abstract}
	\section{Introduction} 
	\cmt{Network structure}Time series are generated and observed in many applications. Using multiple time series data from a complex system, identifying a structure explaining dependencies (connections) among variables is a well-motivated problem in many fields \cite{kolaczyck2009}. Such a networked structure may offer insights about the system dynamics and can assist in inference tasks such as prediction, event detection, and signal reconstruction \cite{liu2016unsupervised},\cite{giannakis2018topology},\cite{zaman2019online}.
	\par 
	\cmt{SEMs}There are different models and approaches that are extensively used in topology identification in certain applications: see, e.g., \cite{mateos2018connecting}, \cite{zaman2017onlinetopology}, \cite{giannakis2018topology}, and references therein. Among these models, structural equation model (SEM) is a popular model \cite{kline2015}: this is mainly due to its tractability and the ability to identify directed relations by means of the inclusion of exogenous variables, which are naturally available in many applications. These exogenous variables represent influences that do not depend on the (endogenous) variables in the model, and their inclusion contributes to the model identifiability \cite {bazerque2013identifiabilitysem}. Static SEMs have been applied to topology identification problems in various fields, e.g., gene regulatory network discovery from gene expression data \cite {cai2013inferencesem}. However, static SEM cannot capture topology changes if the underlying dynamics are nonstationary and each observation is obtained at instants relatively spaced in time, which occurs in various applications. 
	\par 
	\cmt{dynamicSEMs} In time-varying environments, a dynamic SEM can be applied \cite {asparouhov2018dynamicsem}. A dynamic SEM is considered in \cite {baingana2014trackingcascades} to track information cascades of popular news topics over social networks, which are assumed to have sparse dynamic topologies. In the same work, several online algorithms are presented, but not supported by any performance guarantees, so that their tracking capabilities are not theoretically characterized. In \cite{akhavan2017topologytracking}, an online algorithm for tracking dynamic topologies is proposed where the exogenous input is not fully known, also without convergence guarantees. 
	\par  
	\cmt{This paper}In this paper, an online algorithm to track the changes in dynamic SEM topologies in the lines of \cite {baingana2014trackingcascades} is described and its dynamic regret is analyzed, to theoretically characterize its tracking capabilities. The dynamic regret measures the cumulative difference between the cost function evaluated at the estimates and the cost function evaluated at a sequence of time-varying optimal solutions. Specifically, we provide a bound on the dynamic regret that depends on easily measurable properties of the data, the algorithm hyperparameters, and a metric of how much the model varies along time.
	\par
	\cmt{paper structure}The rest of the paper is organized as follows: Sec. \ref {sec:model} contains the model, problem formulation, and the derivation of the algorithm. Sec. \ref {sec:dynamicregret} establishes the dynamic regret bound, including its formal proofs. Numerical results are presented in Sec. \ref {sec:simulations} and Sec. \ref {sec:conclusion} concludes the paper.
	\section{Model and Problem Formulation} \label {sec:model}
		\begin{myitemize}
			\myitem \cmt{Model}
			\begin{myitemize}
				\myitem \cmt{Dynamic linear SEM}Consider a networked system with $N$ nodes, indexed by $i$. At each time frame indexed by $t$, a number $C$ of interactions (frequently denoted as \emph{contagions}) indexed by $c$ are observed in the system, with $y_{ic}^t$ denoting the intensity of the $c$-th contagion in node $i$ at time $t$. Also, let $x_{ic}$ denote the susceptibility of node $i$ to external influence (infection) by contagion $c$. The dynamic linear structural equation model (SEM) is given by \cite{baingana2014trackingcascades}:
				\begin{equation} \label{eq:sem-onenode}
					y_{ic}^t= \sum_{j=1,j \ne i}^{N} a_{ij}^t y_{jc}^t + b_{ii}^t x_{ic} + e_{ic}^t,
				\end{equation}
				for $i=1,\ldots, N, ~ c=1,\ldots,C, ~ t=1,\ldots,T,$ where
				\begin{myitemize}
					\myitem the coefficients $a_{ij}^t$ are the time-varying SEM parameters that encode the topology of the network,					
					\myitem $b_{ii}^t$ quantifies the level of influence of external sources on node $i$,
					\myitem and $e_{ic}$ denotes the measurement errors and un-modeled dynamics.
				\end{myitemize}
				\myitem \cmt{Vector form:}By defining $\bm y_c^t= [y_{1c}^t,\ldots,y_{Nc}^t]^\top \in \mathbb R^{N}$, $\bm x_c=[x_{1c},\ldots,x_{Nc}]^\top\in \mathbb R^{N} $, $\bm B^t= \mathrm{diag}( \bm b^t)\in \mathbb R^{N\times N}$ with $\bm b^t=[b_{11}^t,\ldots,b_{NN}^t]^\top$, and $\bm e_c^t= [e_{1c}^t,\ldots,e_{Nc}^t]^\top \in \mathbb R^{N} $, the model in \eqref{eq:sem-onenode} can also be written in a compact form as:
				\begin{equation}
					\bm y_c^t = \bm A^t \bm y_c ^t + \bm B^t \bm x_c  + \bm e_c^t, ~~~~ c=1, \ldots, C.
				\end{equation}
				The matrix $\bm A^t \in \mathbb R^{N \times N}$ can be seen as a time-varying adjacency matrix for an SEM-based network.
				\myitem \cmt{Matrix form:}The observations for all contagions can be collected in a matrix by defining $\bm Y^t=[\bm y_1^t, \ldots, \bm y_C^t]\in \mathbb R^{N\times C}$, $\bm X=[\bm x_1, \ldots, \bm x_C] \in \mathbb R^{N\times C}$, and $\bm E^t= [\bm e_1^t, \ldots, \bm e_C^t] \in \mathbb R^{N\times C}$. The dynamic SEM takes the following form:
				\begin{equation} \label{eq:SEMmodelmat}
					\bm Y^t= \bm A^t \bm Y^t + \bm B^t \bm X + \bm E^t.
				\end{equation}
			\end{myitemize}
		\myitem \cmt{Problem Statement:}The problem statement becomes: Given the observations $\{ \bm Y^t \}_{t=1}^T$ and $\bm X$, find $\{\bm A^t\}_{t=1}^T$ and $\{ \bm B^t\}_{t=1}^T$.
		\myitem \cmt{Online Criterion}Along the lines of \cite{baingana2014trackingcascades}, we consider the exponentially-weighted least-squares criterion:
		\begin{equation}
		f_t(\bm A, \bm B) \triangleq \frac{1}{2} \sum_{\tau =1}^t \gamma ^{t-\tau} \left \lVert \bm Y^{\tau}- \bm A \bm Y^\tau -\bm B \bm X  \right \rVert _F^2
		\end{equation}
		and the regularizer $\Omega(\bm A) \define \lambda \left \lVert \mathrm{vec}(\bm A) \right \rVert_1,$ and formulate the estimation problem as 
		\begin{subequations} \label{eq:optimization}
			\begin{align}
				\{\hbm A^t, \hbm B^t\}=\,\underset{\bm A, \bm B}{\arg \min} & \; f_t(\bm A, \bm B) + \Omega(\bm A) \\
			 \mathrm {s.to:} & \;a_{ii}=0, \forall i 
			 \\& \;b_{ij}=0, \forall i \ne j. \label{eq:constraints}
			\end{align}
		\end{subequations} 
		\begin{myitemize}
			\myitem The parameter $\gamma \in (0,1]$ is a forgetting factor that regulates how much past information influences the solution at time $t$,
			\myitem and $\lambda$ is the sparsity-promoting regularization parameter.
			\myitem The constraint $a_{ii}=0$ eliminates any component of the trivial solution $\bm A = \bm I$. The constraint $b_{ij}=0$ guarantees a diagonal $\bm B$, meaning that external sources for a certain node $i$ do not affect any other node $j\ne i$. 
		\end{myitemize}	
		Dealing with constraints can be easily avoided if we rewrite the objective including only the nonzero elements of the matrices. We can rewrite $f_t(\bm A, \bm B)$ as:
	\begin{subequations}
		\begin{align} 
			f_t(\bm A, \bm B)&= \frac{1}{2} \sum_{\tau=1}^{t} \sum_{i=1}^{N} \gamma ^{t-\tau} \left \lVert \bm y_i^{\tau\top}-  \bm a_{-i}^\top \bm Y_{-i}^\tau - b_{ii}\bm x_i^\top \right \rVert_F^2 \label{eq:factorized-f}\\
			&= \frac{1}{2} \sum_{\tau=1}^{t} \sum_{i=1}^{N} \gamma ^{t-\tau} \left \lVert \bm y_i^{\tau\top}- [\bm a_{-i}^\top~b_{ii}]   \begin{bmatrix}
				\bm Y_{-i}^\tau \\
				\bm x_i^\top
			\end{bmatrix} \right \rVert_F^2,
		\end{align}
	\end{subequations}
	where $\bm y_i^{\tau\top}$ is the $i$-th row of $\bm Y^\tau$, $\bm x_i^\top$ is the $i$-th row of $\bm X$, $\bm a_{-i}^\top $ is the $i$-th row of $\bm A$ without $i$-th entry, and $\bm Y_{-i}^\tau$ is obtained by removing the $i$-th row from $\bm Y^\tau$. 
	\par
	Further, we can define $\bm v_i\define [\bm a_{-i}^\top ~ b_{ii}]^\top$ and $\bm Z_i^\tau \define [\bm (\bm Y_{-i}^\tau)^\top ~ \bm x_i]^\top$ to rewrite \eqref {eq:factorized-f}:
		\begin{align}
			f_t( \bm A, \bm B)&=  \frac{1}{2} \sum_{\tau=1}^{t} \sum_{i=1}^{N} \gamma ^{t-\tau} \left \lVert \bm y_i^{\tau\top}- \bm v_{i}^\top \bm Z_i^\tau  \right \rVert_F^2\nonumber \\
			&= \frac{1}{2} \sum_{\tau=1}^{t} \sum_{i=1}^{N} \gamma ^{t-\tau} \left \lVert \bm y_i^{\tau}- (\bm Z_i^\tau)^\top \bm v_{i}  \right \rVert_2^2
		\end{align}
	Note that $f_t$ in \eqref {eq:factorized-f} is separable across $i$ (nodes), 
	so that
	\begin{subequations} 
	\begin{align}
	f_t(\bm A, \bm B) &= \sum_{i=1}^{N} f_t^i(\bm v_{i}), \label{eq:separablef} 
	\\
	 \label{eq:deff_ti}
		\text{where} \quad\quad f_t^i(\bm v_{i}) &\define \frac{1}{2} \sum_{\tau=1}^{t}  \gamma ^{t-\tau} \left \lVert \bm y_i^{\tau}- (\bm Z_i^\tau)^\top \bm v_{i}  \right \rVert_2^2. 
	\end{align}
	\end{subequations}
	Similarly, upon defining
	\begin{equation} \label{eq:deff_omega}
		\Omega^i(\bm v_i) \define \lambda \left \lVert \bm a_{-i} \right \rVert_1,
	\end{equation}
	the regularization function is also separable across the rows of $\bm A$, as $\left \lVert \mathrm{vec}(\bm A) \right \rVert_1= \sum _{i=1}^N\Omega^i(\bm v_i)$.

In the next subsection, the online proximal gradient algorithm in \cite{dixit2019onlineproximal} will be applied to solve \eqref{eq:optimization} leveraging the separability we just presented. Before presenting the algorithm, we re-write $f_t^i(\bm v_i)$ in a form that will simplify the computation of its gradient. By expanding \eqref{eq:deff_ti} and ignoring terms not dependent on $\bm v_{i}$:
		\begin{align}
			f_t^i(\bm v_{i}) \propto \frac{1}{2} \sum_{\tau=1}^{t} \! \gamma ^{t-\tau} \Big[
			& \bm v_i^\top \bm Z_i^\tau (\bm Z_i^\tau)^\top \bm v_i \!-\! 2\bm y_i^{\tau\top}  (\bm Z_i^\tau)^\top \bm v_{i} 
			\Big ],\nonumber 
		\end{align}
	the gradient of $f_t^i(\bm v_{i})$ is given by
	\begin{subequations}
		\begin{align} 
			\nabla_{\bm v_i} f_t^i(\bm  v_i) = &\bm \Phi_{\bm Z_i}^t \bm v_i - \bm r_i^t, \label{eq:gradexpression} \\
			\text{where} \;\;\bm \Phi_{\bm Z_i}^t  \define &\sum_{\tau=1}^{t}  \gamma ^{t-\tau} \bm Z_i^\tau (\bm Z_i^\tau)^\top \label{eq:defPhi}\\
			\text{and} \;\;\bm r_i^t \define &\sum_{\tau=1}^{t}  \gamma ^{t-\tau} \bm Z_i^\tau (\bm y_i^{\tau \top})^\top. \label{eq:defr}
		\end{align}
	\end{subequations}
	Note that the variables defined in the latter two expressions can be computed recursively, as will be expressed in the tabulated algorithm (lines 5 and 6).
		
	\subsection{Proximal online gradient algorithm}
	The update of the proximal online gradient descent algorithm \cite{dixit2019onlineproximal}, applied to the $i$-th portion of the separable problem presented in the previous section, yields
		\begin{equation} \label{eq:proxupdate}
			\bm v_i[t+1]= \textbf {prox}_{\Omega_i}^\alpha \left(\bm g_i^\alpha[t](\bm v_i[t])\right ),
		\end{equation}
		where
		\begin{myitemize}
			\myitem $\alpha>0$,
			\myitem ${\bm g}_i^{\alpha}[t](\bm u) \define \bm u - \alpha \nabla _{\bm u} f_i^t(\bm u)$, and
			\myitem  
			\begin{equation} \label {eq:defproximaloperator}
				\textbf{prox}_{ \Psi}^\alpha (\bm w)\triangleq \underset{\bm s \in \mathrm{dom }\Psi}{\arg\min}\left [\Psi(\bm s)+\frac{1}{2 \alpha }\left \lVert \bm s-\bm w\right \rVert_2^2\right].
			\end{equation}
		\end{myitemize}
	
	From the definition of $\bm v_i$ and \eqref{eq:deff_omega}, it becomes clear that 
	\begin{equation}
	\textbf{prox}_{\Omega_i}^\alpha (\bm s) = \left[ S_{\alpha\lambda}([\bm s]_{1:N-1})^\top [\bm s]_N\right]^\top
	\end{equation}
	with $S_{\alpha\lambda}(\bm w) 
	$ denoting the standard soft-thresholding operator.
	The complete procedure is presented in \textbf{Algorithm \ref {alg:dynamicSEM}}. Observe that the step size $\alpha$ is required to be small enough, specifically $\alpha < 1/L_f $ where $\lambda_{\mathrm {max}}(\bm \Phi_{\bm Z_i}^t) \leq L_f, ~\forall~ i,t$.

	\begin{algorithm}
		\caption{Online algorithm for tracking dynamic SEM-based Topologies}\label{alg:dynamicSEM}
		\textbf{Input:} $\gamma, \lambda, \alpha \in (0,1/L_f], \{\bm Y^t\}_{t=1}^{T}, \bm X$ \\
		\textbf{Output:} $\{\bm A[t]\}_{t=1}^T$,  $\{ \bm B[t]\}_{t=1}^T$ \\
		\textbf{Initialization:} $  \bm v_i[1]=0_{N\times 1}, \bm \Phi_{\bm Z_i}^0\!=\! \bm 0_{N\times N}, ~\bm r_i^0=\bm 0_{N\times1}, i=1,\ldots, N $
		\begin{algorithmic}[1]
			\For {$t=1,2,\ldots,T $} 
			\State {Receive data  $\bm Y^t$}
			\For {$i=1,2,\ldots,N$}
				\State{$\bm Z_i^t = [\bm (\bm Y_{-i}^t)^\top ~ (\bm x_i^\top)^\top]^\top$}
				\State $\bm \Phi_{\bm Z_i}^t= \gamma \, \bm \Phi_{\bm Z_i}^{t-1} + \bm Z_i (\bm Z_i)^\top$ 
				\State $\bm r_i^t= \gamma \, \bm r_i^{t-1}+\bm Z_i^t (\bm y_i^{t \top})^\top$ 
				\State $\nabla_{\bm v_i} f_t^i(\bm  v_i[t])= \bm \Phi_{\bm Z_i}^t \bm v_i[t] - \bm r_i^t$ 
				\State  $\bm  v_i^{\text{f}}[t] = \bm  v_i[t]- \alpha \nabla_{\bm v_i} f_t^i(\bm  v_i[t])$
				\State {$\bm  v_i[t+1]= \textbf{prox}_{ \Omega ^i}^\alpha \left (\bm  v_i^{\text{f}}[t]\right )$ }
				\EndFor
				\State \textbf{end for}
				\State Form $\bm A[t]$ and $\bm B[t]$ from $\bm v_i[t], i=1,...,N$
			\EndFor
			\State \textbf{end for}
		\end{algorithmic}
	\end{algorithm}
\end{myitemize}
\section{Dynamic Regret Analysis}\label{sec:dynamicregret}
	\begin{myitemize} 	
	\myitem \cmt{Dynamic regret derivation}The performance of online algorithms is evaluated by means of the regret, which is the difference in performance between the online algorithm and a solution which can be computed based on the data in hindsight. The regret measure can be static or dynamic. In the case of the static regret, the best comparator minimizes the objective averaged over all past instants, which implicitly assumes a stationary model. Therefore, the static regret cannot express the tracking performance of an online algorithm in dynamic environments, where the generating parameters are time-varying. To characterize the tracking performance of online algorithms, the dynamic regret \cite{hall2015dynamicregret} is used, which results from comparing the online algorithm against an optimal sequence of time-varying hindsight solutions. Specifically, upon defining $h_t(\bm A[t], \bm B[t]) \define f_t(\bm A[t], \bm B[t]) + \Omega(\bm A[t])$, 
 the dynamic regret is given by:
	\begin{equation}
		R_d[T]= \sum_{t=1}^{T} \left[ h_t(\bm A[t], \bm B[t]) - h_t(\bm A^\star[t], \bm B^\star[t]) \right ].
	\end{equation}
	with $(\bm A^\star[t], \bm B^\star[t])$ representing the estimate produced by a clairvoyant that knows $h_t(\cdot)$ in advance (in contrast, the online algorithm does not have access to $h_t(\cdot)$ while producing $(\bm A[t], \bm B[t])$).
	Using \eqref  {eq:factorized-f}, the above expression can be written as:
	\begin{align*}
		R_d[T] &= \sum_{t=1}^{T} \sum_{i=1}^{N} \left[ f_t^i(\bm v_i[t])\!+ \!\Omega^i(\bm v_i[t])\!-\! f_t^i(\bm v_i^\star[t]) \!- \!\Omega^i (\bm v_i^\star[t]) \right ]\\
		&= \sum_{t=1}^{T} \sum_{i=1}^{N} \left[ h_t^i(\bm v_i[t])-h_t^i(\bm v_i^\star[t]) \right ] = \sum_{i=1}^{N} R_d^i[T],
	\end{align*}
	where 
	$h_t^i(\bm v_i[t]) \define f_t^i(\bm v_i[t])+ \Omega^i(\bm v_i[t])$, 
	$\bm v_i^\star [t] \triangleq \arg \min_{\bm v_i} f_t^i(\bm v_i)+\Omega^i (\bm v_i)$, and $R_d^i[T] \define \sum_{t=1}^{T} [ h_t^i(\bm v_i[t])- h_t^i(\bm v_i^\star[t]) ]$. 
	Observe that the regret expression is separable across index $i$ (nodes). Thus, for the sake of simplicity, we derive the regret for the $i$-th node, i.e., $R_d^i[T]$. The total regret will be obtained by adding the individual regret expressions. 
	We define the path length for each subproblem (corresponding to each node $i$) as:
	\begin{equation}
		W_i[T] \define \sum_{t=2}^T \left \lVert \bm v_i^\star [t]-\bm v_i^\star [t-1] \right \rVert_2,
	\end{equation}
	which represents the aggregated variations in the consecutive optimal solutions. 
	\myitem \cmt{Assumptions}In this work, the following assumptions are considered:
	\begin{enumerate}[{A}1.]
		\item{ \textit{Bounded process:} There exists $B_{xy}$ such that $|y_{ic}^t|^2 \leq B_{xy} $ and $|x_{ic}|^2 \leq B_{xy} $, ~ $\forall~ i,c,t$.} \label{as:boundedprocess}
		\item { \textit{Strong convexity:} Each function $f_t^i$ is $\beta$-strongly convex, i.e., $\lambda_{\mathrm {min}} (\bm \Phi_{\bm Z_i}^t) \geq \beta >0, ~\forall~ i,t$.} \label{as:storngconvexity}
		\item { \textit{Lipschitz smoothness:} Each function $f_t^i$ is $L_f$-Lipschitz smooth, i.e., $\lambda_{\mathrm {max}}(\bm \Phi_{\bm Z_i}^t) \leq L_f, ~\forall~ i,t$.} \label{as:lipschitzsmoothness}
		\item { \textit{Bounded variations of the optimal solution:} The distance between two consecutive optimal solution is bounded, i.e.,
			\begin{equation}
			\left \lVert \bm v_i^\star[t] - \bm v_i^\star[t+1]\right \rVert_2 \leq d, d\geq 0,\forall \,t,i.
			\end{equation}	}\label{as:boundedvariations}
	\end{enumerate}
	These above assumptions are standard in the literature. Assumption A\ref{as:boundedprocess} does not entail any loss of generality and is satisfied in most real-world applications.  
	Next, we present an upper bound on the dynamic regret. 
	\myitem \cmt{Theorem}
	\begin{theorem}
		The individual dynamic regret of \textbf{Algorithm \ref{alg:dynamicSEM}} for a node $i$ is given by:
		\begin{equation}
		R_d^i[T]= D_h \left (\left \lVert \bm v_{i}^\star[1] \right \rVert_2 + W_i[T] \right),
		\end{equation}
		where 
		\begin{equation} \label{eq:regretconstant}
			D_h \define \frac{1}{\alpha \beta }  \left ( \frac{B_{xy}C\sqrt{N}}{1-\gamma} \left (1+\frac{L_f}{\beta}\right )+ \lambda \sqrt{N-1}\right ),
		\end{equation}
		under assumptions A1, A2, A3, and A4.
	\end{theorem}
\begin{proof}
	Since $h_t^i$ is convex, we have by definition that:
	\begin{equation}
		h_t^i(\bm v_i^\star[t]) \geq h_t^i(\bm v_i[t]) + (\tilde \nabla h_t^i(\bm v_i[t]))^\top  (\bm v_i^\star[t]-\bm v_i[t]),
	\end{equation}
	$\forall \,\bm v_i[t], \bm v_i^\star[t]$,  $\tilde \nabla h_t^i(\bm v_i[t])$ denotes a subgradient of $h_t^i(\bm v_i[t])$ given by $\tilde \nabla h_t^i(\bm u)=  \nabla f_t^i(\bm u) + \tilde \nabla \Omega^i(\bm u)$ with $\tilde \nabla \Omega^i(\bm u) \in \partial \Omega^i(\bm u)$. 
	Rearranging and summing the above inequality from $t=1$ to $T$, we have
	\begin{align}\label{eq:regretexp}
		\sum_{t=1}^T \left [ h_t^i(\bm v_i[t]) \;-\right.& \left. h_t^i(\bm v_i^\star )\right ] 
		\leq \sum_{t=1}^T \left(\tilde \nabla h_t^i(\bm v_i[t]) \right )^\top\!(\bm v_i[t] -\bm v_i^\star[t])\nonumber \\
		\leq &\sum_{t=1}^T  \left \lVert\tilde \nabla h_t^i(\bm v_i[t]) \right \rVert_2 \cdot \left \lVert\bm v_i[t] -\bm v_i^\star[t]\right \rVert_2,
	\end{align} 
	where the second inequality follows from the Cauchy-Schwarz inequality.
	Next, we derive a bound on $ \lVert\tilde \nabla h_t^i(\bm v_i[t])  \rVert_2$. Note first that it holds that
		\begin{equation}
			\left \lVert \tilde \nabla (h_t^i(\bm v_i[t])\right \rVert_2
			\leq \left \lVert  \nabla f_t^i(\bm v_i[t]) \right \rVert_2 +\left \lVert \tilde \nabla\Omega^i(\bm v_i[t]) \right \rVert_2. \label{eq:boundonsubgrad}
		\end{equation}
		Thus, we have to prove that $\lVert  \nabla f_t^i(\bm v_i[t])  \rVert_2 $ and $ \lVert \tilde \nabla\Omega^i(\bm v_i[t])  \rVert_2$ are bounded $\forall \, \bm v_i[t]$. First, we prove that $\lVert  \nabla f_t^i(\bm v_i[t])  \rVert_2 $ is bounded. To this end, from \eqref  {eq:gradexpression}, using the triangular inequality, the spectral radius of $\bm \Phi_{\bm Z_i}^t$, and assumption  A\ref{as:lipschitzsmoothness}, we obtain the following:
			\begin{align}
				\left \lVert  \nabla f_t^i(\bm v_i[t+1])  \right\rVert_2 
				& =\left \lVert  \bm \Phi_{\bm Z_i}^t \bm v_i [t+1]- \bm r_i^t  \right\rVert_2\nonumber \\
				&\leq  \left \lVert  \bm \Phi_{\bm Z_i}^t \bm v_i [t+1]\right\rVert_2+\left \lVert  \bm r_i^t  \right\rVert_2\nonumber \\
				&\leq  \lambda_{\mathrm {max}}( \bm \Phi_{\bm Z_i}^t)\left \lVert  \bm v_i [t+1]\right\rVert_2\!+\!\left \lVert  \bm r_i^t  \right\rVert_2 \nonumber \\
				&\leq L_f\left \lVert  \bm v_i [t+1]\right\rVert_2+\left \lVert  \bm r_i^t  \right\rVert_2. \label{eq:gradboundinVandr}
			\end{align}
		We need to derive a bound on $\lVert  \bm v_i [t+1]\rVert_2$ and $\lVert \bm r_i^t \rVert_2$. First, we derive a bound on $\lVert \bm r_i^t \rVert_2$. From the definition of $ \bm r_i^t$ in \eqref  {eq:defr}, and using assumption A\ref{as:boundedprocess}, we obtain the bound as follows:  
			\begin{align}
				\left \lVert  \bm r_i^t  \right\rVert_2 &=  \left \lVert \sum_{\tau=1}^{t}  \gamma ^{t-\tau} \bm Z_i^\tau (\bm y_i^{\tau \top})^\top\right\rVert_2 \leq \left \lVert \sum_{\tau=1}^{t}  \gamma ^{t-\tau} B_{xy} \bm 1_{N\times C}\bm 1_{C}\right\rVert_2\nonumber\\
				& = B_{xy}\left \lVert \sum_{\tau=1}^{t}  \gamma ^{t-\tau}  C\bm 1_{N}\right\rVert_2 = B_{xy}C \sum_{\tau=1}^{t}  \gamma ^{t-\tau}  \left \lVert\bm 1_{N}\right\rVert_2 \nonumber\\
				& \leq \frac{B_{xy}C\sqrt{N}}{1-\gamma} = \frac{B_{xy}C\sqrt{N}}{\mu}, \label{eq:boundonr}
			\end{align}
		where $\mu \define 1- \gamma$. Thus, we have derived a bound on  $\left \lVert  \bm r_i^t  \right\rVert_2$.
		To derive an upperbound on $\lVert  \bm v_i [t+1]\rVert_2$, from the update expression of the algorithm, and using assumption A\ref{as:storngconvexity}, we have that:
			\begin{align}
				\left \lVert \bm v_i[t+1] \right  \rVert_2 & \leq 
			 \left \lVert \bm v_i[t] - \alpha \nabla f_i^t(\bm v_i[t])\right  \rVert_2 \nonumber \\
				& = \left \lVert \bm v_i[t] - \alpha \left (\bm \Phi_{\bm Z_i}^t \bm v_i[t] -\bm r_i^t \right )\right  \rVert_2 \nonumber \\
				& =\left \lVert \left ( \bm I - \alpha \bm \Phi_{\bm Z_i}^t\right)\bm v_i[t] + \alpha \bm r_i^t \right  \rVert_2 \nonumber \\
				& \leq  \lambda_{\mathrm{max}} \left ( \bm I - \alpha \bm \Phi_{\bm Z_i}^t\right) \left \lVert \bm v_i[t] \right \rVert_2+ \alpha \left \lVert \bm r_i^t\right  \rVert_2 \nonumber\\
				& =\left ( 1 - \alpha \lambda_{\mathrm{min}} ( \bm \Phi_{\bm Z_i}^t\right) \left \lVert \bm v_i[t] \right \rVert_2 + \alpha \left \lVert \bm r_i^t \right  \rVert_2 \nonumber \\
				&\leq \left ( 1 - \alpha\beta \right )\left \lVert \bm v_i[t] \right \rVert_2 + \alpha \left \lVert \bm r_i^t \right  \rVert_2. \label{eq:boundonestimates}
			\end{align}
		Substituting the bound on $\bm r_i^t$ from \eqref {eq:boundonr} in the above inequality, we obtain:
			\begin{align*}
				\left \lVert \bm v_i[t+1] \right  \rVert_2 & \leq  ( 1 - \alpha\beta  )\left \lVert \bm v_i[t] \right \rVert_2 + \alpha \frac{B_{xy}C\sqrt{N}}{\mu}\\
				& = \delta \left \lVert \bm v_i[t] \right \rVert_2 +  \frac{\alpha B_{xy}C\sqrt{N}}{\mu}.
			\end{align*}
		By recursive substitution in the above inequality:
		\begin{subequations}
			\begin{align*}
			\left \lVert \bm v_i[t+1] \right  \rVert_2 & \leq \delta \left ( \! \delta \left \lVert \bm v_i[t\!-\!1] \right \rVert_2 \! + \!  \frac{\alpha B_{xy}C\sqrt{N}}{\mu}\right) \!+ \! \frac{\alpha B_{xy}C\sqrt{N}}{\mu}\\
			& = \delta^2 \left \lVert \bm v_i[t-1] \right \rVert_2 + \delta \frac{\alpha B_{xy}C\sqrt{N}}{\mu} + \frac{\alpha B_{xy}C\sqrt{N}}{\mu}\\
			& \leq  \delta^3 \left \lVert \bm v_i[t-2] \right \rVert_2 + \frac{\alpha B_{xy}C\sqrt{N}}{\mu}( \delta^2 \!+\! \delta \! +\! 1) \leq \ldots \\
			& \leq \delta ^{k} \left \lVert \bm v_i[t-k+1] \right \rVert_2 + \frac{\alpha B_{xy}C\sqrt{N}}{\mu} \sum_{i=0}^{k-1} \delta^i,
			\end{align*}
		\end{subequations}
		where $1\leq k\leq t$. For $k=t$, the above inequality becomes
			\begin{align}
				\left \lVert \bm v_i[t+1] \right  \rVert_2 & \leq \delta ^{t} \left \lVert \bm v_i[+1] \right \rVert_2 + \frac{\alpha B_{xy}C\sqrt{N}}{\mu} \sum_{i=0}^{t-1}\delta^i \nonumber \\
				& = \frac{\alpha B_{xy}C\sqrt{N}}{\mu} \frac{1-\delta ^t}{1-\delta}  \leq \frac{\alpha B_{xy}C\sqrt{N}}{\mu} \frac{1}{\alpha \beta}	\nonumber\\
				& = \frac{ B_{xy}C\sqrt{N}}{\mu \beta}. \label{eq:boundonestiamtes}
			\end{align}
		By substituting the bounds from \eqref {eq:boundonestiamtes} and \eqref {eq:boundonr} into  \eqref {eq:gradboundinVandr}, we obtain the bound on $\left \lVert  \nabla f_t^i(\bm v_i[t+1])  \right\rVert_2$ as follows: 
		\begin{subequations}
			\begin{align}
				\left \lVert  \nabla f_t^i(\bm v_i[t+1])  \right\rVert_2 & \leq  \frac{ L_f B_{xy}C\sqrt{N}}{\mu \beta} + \frac{B_{xy}C\sqrt{N}}{\mu}\\
				& = \frac{B_{xy}C\sqrt{N}}{\mu} \left (1+\frac{L_f}{\beta}\right ). \label{eq:gradbound}
			\end{align} 
		\end{subequations}
		To prove that $\lVert \tilde \nabla\Omega^i(\bm v_i[t]) \rVert_2$ is bounded $\forall \,\bm v_i[t]$, first we compute the Lipschitz continuity parameter of $\Omega^i$, i.e., $L_{\Omega}$ and then apply the result in  \cite [Lemma 2.6]{shalev2011online}, which establishes that all the subgradients of a function are bounded by its Lipschitz continuity parameter. To find $L_{\Omega}$, let $\bm a'\define [\bm a^\top m]^\top, \bm b' \define [\bm b^\top n]^\top, \bm a,\bm b \in \mathbb R^{N-1},~m,n \in \mathbb R$. By the triangular inequality and the reverse triangular inequality, we have that:
			\begin{align*}
				\left\lvert \Omega^i(\bm a') -\Omega^i(\bm b') \right\rvert&= \left\lvert \lambda \left \lVert\bm a \right \rVert_1 -\lambda \left \lVert \bm b \right \rVert_1\right\rvert\\
				&= \left\lvert \lambda \sum_{i=1}^{N-1}\left [|a_i|-|b_i|\right ]\right\rvert\\
				&\leq \lambda \sum_{i=1}^{N-1} \big\lvert|a_i|-|b_i|\big\rvert
				\leq \lambda \sum_{i=1}^{N-1} \lvert a_i-b_i \rvert \\
				&=\lambda \left \lVert \bm a - \bm b \right \rVert_1\leq \lambda \sqrt{N-1}\left \lVert \bm a - \bm b \right \rVert_2\\
				&\leq \lambda \sqrt{N-1}\left \lVert \bm a' - \bm b' \right \rVert_2.
			\end{align*}
	Thus, we have that $L_{\Omega}=\lambda \sqrt{N-1}$. Substituting these bounds in \eqref {eq:boundonsubgrad}, we have
	\begin{align}
		\left \lVert \tilde \nabla (f_t^i(\bm v_i)\right \rVert_2 \! \!\!+\! \left \lVert\tilde \nabla \Omega^i(\bm v_i)) \right \rVert_2\! \leq \! \frac{B_{xy}C\sqrt{N}}{\mu} \left (\!1\!  +\! \frac{L_f}{\beta}\right ) \!+ \!\lambda \sqrt{N\!-\!1}.
	\end{align}
	Substituting the above bound in \eqref {eq:regretexp}, we obtain
	\begin{align}
		& \sum_{t=1}^T \left [ h_t^i(\bm v_i[t])- h_t^i(\bm v_i^\star )\right ] \nonumber \\&\leq \sum_{t=1}^T  \Bigg (\frac{B_{xy}C\sqrt{N}}{\mu} \left (1+\frac{L_f}{\beta} \right ) 
		+ \lambda \sqrt{N-1}\Bigg ) \left \lVert\bm v_i[t] -\bm v_i^\star[t]\right \rVert_2\nonumber\\
		&=   \left (\frac{B_{xy}C\sqrt{N}}{\mu} \left (1\!+\!\frac{L_f}{\beta}\right )\!+\! \lambda \sqrt{N\!- \!1}\right) \sum_{t=1}^T\left \lVert\bm v_i[t] -\bm v_i^\star[t]\right \rVert_2. \label{eq:dregretintermsofv}
	\end{align}
	Next, we derive a bound on $\sum_{t=1}^T\left \lVert\bm v_i[t] -\bm v_i^\star[t]\right \rVert_2$. To this end, the first step is to prove the following result:
	\begin{equation} \label{eq:boundingdifference}
		\left \lVert \bm v_i[t+1] - \bm v_i^\star [t]\right \rVert_2 \leq \rho \left \lVert \bm v_i[t] - \bm v_i^\star [t]\right \rVert_2,
	\end{equation}
	where $\rho= 1-\alpha \beta$.
	To this end, squaring the l.h.s. of \eqref {eq:boundingdifference} and by definition of $\bm v_i[t+1]$, we have
	\begin{subequations}
		\begin{align*}
		&\left \lVert \bm v_i[t+1] - \bm v_i^\star [t]\right \rVert_2^2
		 = \left \lVert \begin{bmatrix}
		\bm a_{-i}[t+1]- \bm a_{-i}^\star [t]\\
		b_{ii}[t+1]- b_{ii}^\star [t]
		\end{bmatrix}\right \rVert_2^2\\
		& = \left \lVert \bm a_{-i}[t+1]- \bm a_{-i}^\star [t]\right \rVert_2^2 + \left \lVert b_{ii}[t+1]- b_{ii}^\star [t]\right \rVert_2^2\\
		& = \big \lVert \textbf{prox}_{ \lambda \lVert \cdot \rVert_1}^\alpha (\bm a_{-i}[t]- \alpha \nabla_{\bm a_{-i}}f_t^i(\bm a_{-i}[t])) \\
		& \quad  - \textbf{prox}_{ \lambda \lVert \cdot \rVert_1}^\alpha (\bm a_{-i}^\star[t]- \alpha \nabla_{\bm a_{-i}}f_t^i(\bm a_{-i}^\star [t])) \big \rVert_2^2 \nonumber  \\
		& \quad + \left ( b_{ii}[t]- \alpha \nabla_{b_{ii}}f_t^i(b_{ii}[t]) -( b_{ii}^\star [t]-\alpha \nabla_{b_{ii}}f_t^i(b_{ii}^\star [t])) \right)^2\\
		& \leq \left \lVert  (\bm a_{-i}[t]\! -\! \alpha \nabla_{\bm a_{-i}}f_t^i(\bm a_{-i}[t])) \! - \!  (\bm a_{-i}^\star[t]\! -\!  \alpha \nabla_{\bm a_{-i}}f_t^i(\bm a_{-i}^\star [t])) \right \rVert_2^2 \nonumber  \\
		& \quad + \left ( b_{ii}[t]- \alpha \nabla_{b_{ii}}f_t^i(b_{ii}[t]) -( b_{ii}^\star [t]-\alpha \nabla_{b_{ii}}f_t^i(b_{ii}^\star [t])) \right)^2\\
		& = \left \lVert  (\bm v_{i}[t]- \alpha \nabla_{\bm v_{i}}f_t^i(\bm v_{i}[t])) -  (\bm v_{i}^\star[t]- \alpha \nabla_{\bm v_{i}}f_t^i(\bm v_{i}^\star [t])) \right \rVert_2^2\\
		&= \left \lVert  \bm v_{i}[t]- \bm v_{i}^\star[t] \right \rVert_2^2 +\alpha ^2 \left \lVert \nabla_{\bm v_{i}}f_t^i(\bm v_{i}[t]) -   \nabla_{\bm v_{i}}f_t^i(\bm v_{i}^\star [t]) \right \rVert_2^2  \nonumber \\
		& \quad -2\alpha (\bm v_{i}[t]- \bm v_{i}^\star[t] )^\top (\nabla_{\bm v_{i}}f_t^i(\bm v_{i}[t]) -   \nabla_{\bm v_{i}}f_t^i(\bm v_{i}^\star [t])) \\
		&\leq \left \lVert  \bm v_{i}[t]- \bm v_{i}^\star[t] \right \rVert_2^2 +\alpha ^2 \left \lVert \nabla_{\bm v_{i}}f_t^i(\bm v_{i}[t]) -   \nabla_{\bm v_{i}}f_t^i(\bm v_{i}^\star [t]) \right \rVert_2^2 \\
		&\quad -2\alpha \Big(\frac{\beta L_f}{L_f+\beta}\left \lVert  \bm v_{i}[t]- \bm v_{i}^\star[t] \right \rVert_2^2 \nonumber \\
		& \quad + \frac{1}{L_f+\beta} \left \lVert \nabla_{\bm v_{i}}f_t^i(\bm v_{i}[t]) -   \nabla_{\bm v_{i}}f_t^i(\bm v_{i}^\star [t]) \right \rVert_2^2 \Big),
		\end{align*}
	\end{subequations}
where the above inequality is implied by assumptions A\ref{as:storngconvexity} and A\ref{as:lipschitzsmoothness}. Given that $\alpha \in (0,1/L_f]$, using assumption A\ref {as:storngconvexity}, and by further simplifications, we have:
	\begin{subequations}
	\begin{align*}
	&\left \lVert \bm v_i[t+1] - \bm v_i^\star [t]\right \rVert_2^2\\
	&\leq \left (1\! -\! \frac{2\alpha \beta L_f}{L_f\! +\! \beta}\right) \left \lVert  \bm v_{i}[t]\! -\!  \bm v_{i}^\star[t] \right \rVert_2^2 \! +\! \left (\alpha ^2 \! -\! \frac{2\alpha}{L_f\! +\! \beta}\right )\cdot \\
		&\quad \big \lVert \nabla_{\bm v_{i}}f_t^i(\bm v_{i}[t]) - \nabla_{\bm v_{i}}f_t^i(\bm v_{i}^\star [t]) \big \rVert_2^2 \\
		&=\left (1-\frac{2\alpha \beta L_f}{L_f+\beta}\right) \left \lVert  \bm v_{i}[t]- \bm v_{i}^\star[t] \right \rVert_2^2 -\left (\frac{2\alpha}{L_f+\beta}-\alpha ^2 \right ) \cdot\\
		& \quad \big \lVert \nabla_{\bm v_{i}}f_t^i(\bm v_{i}[t])- \nabla_{\bm v_{i}}f_t^i(\bm v_{i}^\star [t]) \big \rVert_2^2 \\
		&\leq \left (1\!-\!\frac{2\alpha \beta L_f}{L_f+\beta}\right) \left \lVert  \bm v_{i}[t]- \bm v_{i}^\star[t] \right \rVert_2^2 \!-\!\beta^2\left (\frac{2\alpha}{L_f+\beta}\!-\!\alpha ^2 \right ) \cdot \\
		& \quad \left \lVert  \bm v_{i}[t]- \bm v_{i}^\star[t] \right \rVert_2^2\\
		&= \left (1-\frac{2\alpha \beta L_f}{L_f+\beta}-\frac{2\alpha\beta^2}{L_f+\beta}+\alpha ^2\beta^2\right) \left \lVert  \bm v_{i}[t]- \bm v_{i}^\star[t] \right \rVert_2^2 \\
		&= \left (1-2\alpha \beta +\alpha ^2\beta^2\right) \left \lVert  \bm v_{i}[t]- \bm v_{i}^\star[t] \right \rVert_2^2 \\
		&= \rho^2 \left \lVert  \bm v_{i}[t]- \bm v_{i}^\star[t] \right \rVert_2^2,
		\end{align*}
	\end{subequations}
	where $\rho \define 1-\alpha \beta $. Taking square root on both sides of the above inequality yields \eqref {eq:boundingdifference}. 
	\par 
	Next, we show that
	\begin{equation} \label {eq:cumulaivegap}
	\sum_{t=1}^T \left \lVert \bm v_{i}[t]- \bm v_{i}^\star[t] \right \rVert_2 \leq \frac{1}{1- \rho} \left [\left \lVert \bm v_{i}[1]- \bm v_{i}^\star[1] \right \rVert_2 + W_i[T] \right].
	\end{equation}
	To prove the above expression, consider the cumulative gap 
	\begin{subequations} 
		\begin{align*}
			&\sum_{t=2}^T \left \lVert \bm v_{i}[t] \!-\! \bm v_{i}^\star[t] \right \rVert_2 \!= \!\sum_{t=2}^T \left \lVert \bm v_{i}[t] \! - \!\bm v_{i}^\star[t\!-\!1]\!+\!\bm v_{i}^\star[t\!-\!1]\!-\!\bm v_{i}^\star[t] \right \rVert_2 \\
			& \quad \quad \leq \sum_{t=2}^T \left [\left \lVert \bm v_{i}[t]- \bm v_{i}^\star[t-1]\right \rVert_2+\left \lVert \bm v_{i}^\star[t] - \bm v_{i}^\star[t-1] \right \rVert_2 \right  ]\\
			& \quad \quad = \sum_{t=1}^{T-1} \left \lVert \bm v_{i}[t+1]- \bm v_{i}^\star[t]\right \rVert_2+W_i[T] \\
			& \quad \quad \leq \sum_{t=1}^{T-1} \rho\left \lVert \bm v_{i}[t]- \bm v_{i}^\star[t]\right \rVert_2+W_i[T] \\
			& \quad \quad \leq \sum_{t=1}^{T} \rho\left \lVert \bm v_{i}[t]- \bm v_{i}^\star[t]\right \rVert_2+W_i[T].
	\end{align*}
	\end{subequations}
	Adding $ \lVert \bm v_{i}[1]- \bm v_{i}^ \star[1] \rVert_2$ on both sides of the above inequality results in:
	\begin{align}
		&\sum_{t=1}^T \left \lVert \bm v_{i}[t]- \bm v_{i}^\star[t] \right \rVert_2 \leq  \nonumber \\
		&\sum_{t=1}^{T} \rho\left \lVert \bm v_{i}[t]- \bm v_{i}^\star[t]\right \rVert_2 
		+\left \lVert \bm v_{i}[1]- \bm v_{i}^\star[1] \right \rVert_2+W_i[T].
	\end{align}
	By rearranging terms in the above inequality, we obtain the result in  \eqref {eq:cumulaivegap}. Thus, we can substitute $\sum_{t=1}^T  \lVert \bm v_{i}[t]- \bm v_{i}^\star[t] \rVert_2$ with its bound from \eqref {eq:cumulaivegap} into \eqref{eq:dregretintermsofv} and note that $v_i[1]=0_{N \times 1}$ in \textbf{Algorithm \ref{alg:dynamicSEM}}. This completes the proof.
	\end{proof}
\textbf{Remarks.} The bound on the total dynamic regret is given by: 
\begin{equation}
R_d[T]= D_h \sum_{i=1}^N\left (\left \lVert \bm v_{i}^\star[1] \right \rVert_2 + W_i[T] \right),
\end{equation}
where $D_h$ is defined in \eqref {eq:regretconstant}. Notice that this means that the bound on the dynamic regret is a function of the parameters of the data and the parameters of the algorithm. Moreover, for a sublinear path length, the dynamic regret of the proposed algorithm is sublinear.  
\end{myitemize}
\section{Numerical Results} \label {sec:simulations}
In this section, the performance of the algorithm is analyzed by presenting numerical tests. The experimental results are based on synthetic data. 
\par
To generate the matrices $\bm A^t$, a binary adjacency matrix $\bm A_{\text{binary}}$ is generated according to an Erd\H{o}s-R\'enyi model with edge probability $p_e$. No self-loops are considered, i.e., the diagonal entries of $\bm A_{\text{binary}}$ are zero.
Two models are considered in the simulations: a) smooth-transition model  and b) non-smooth transition model. In the smooth-transition model, the nonzero elements of $\bm A^t$ follow the pattern of the 1's in $\bm A_{\text{binary}}$. For $t=1$, the nonzero elements of $\bm A^t$ take one of the following four functions via random selection: i) $a1(t)= 0.5+0.5\sin(0.1t)$, ii) $a2(t)= 0.5+0.5\cos(0.1t)$, iii) $a3(t)= \exp (-0.01t)$, and iv) $a4(t)=0$. For $t>1$, the elements of $\bm A^t$ evolve according to the function selected initially by evaluating the functions for $t$. In the non-smooth transition model, a static model for $\bm A^t$ is considered for $t<T/2$. The nonzero elements of $A^t$ are drawn one  time from a standard Gaussian distribution. At $t=T/2$, the model changes from one to another.  In both models, the matrices $\bm B^t$ are assumed to be constant, i.e., $\bm B^t= \mathrm{diag}(\bm b)$, where $\bm b$ is fixed and chosen one time randomly from a standard Gaussian distribution.
This assumption means that the coefficients of external influences are constant over time, which is natural since $\bm X$ is constant in the model (cf. \eqref{eq:SEMmodelmat}). At each time $t$, for each contagion $c$, $\bm e_c^t$ is drawn from $ \mathcal N(\bm 0_{N\times 1}, \sigma \bm I_{N \times N})$. At time $t$, $\bm Y^t$ is generated using \eqref {eq:SEMmodelmat}.  
\par Fig. \ref {fig:mse} presents the mean-square error (MSE) given by
$1/N^2\sum_{i=1}^N \lVert  \bm v_i[t]- \bm v_i^{\text{true}}[t] \rVert_2^2$ versus time, and Fig. \ref{fig:dregret} shows the dynamic regret $R_d[T]$ for both models. Since the optimal solution is time-varying, the algorithm is required to track the changes in the optimal solutions. Observe from Fig. \ref {fig:mse} that the MSE has a decreasing trend, meaning that the proposed algorithm is able to track the changes in the time-varying topologies. Fig. \ref{fig:dregret} shows that the dynamic regret of the non-smooth (single breaking point) transition model is lower than that of the smooth-transition model, since the model is always changing in the smooth-transition model. 
\begin{figure}
	\begin{subfigure}{0.5\textwidth}
		\includegraphics[width=\textwidth]{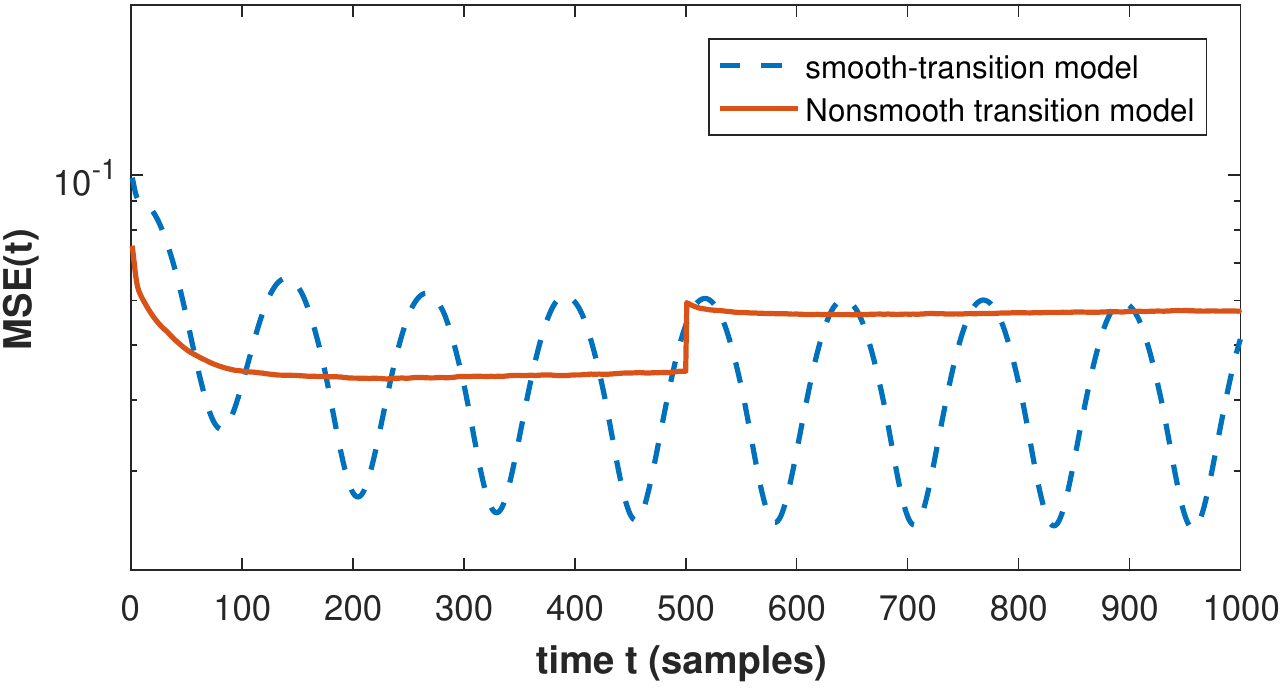}
		\caption{MSE vs. time $t$}
		\label{fig:mse}
		\end{subfigure}
	
		\begin{subfigure}{0.5\textwidth}
				\includegraphics[width=\textwidth]{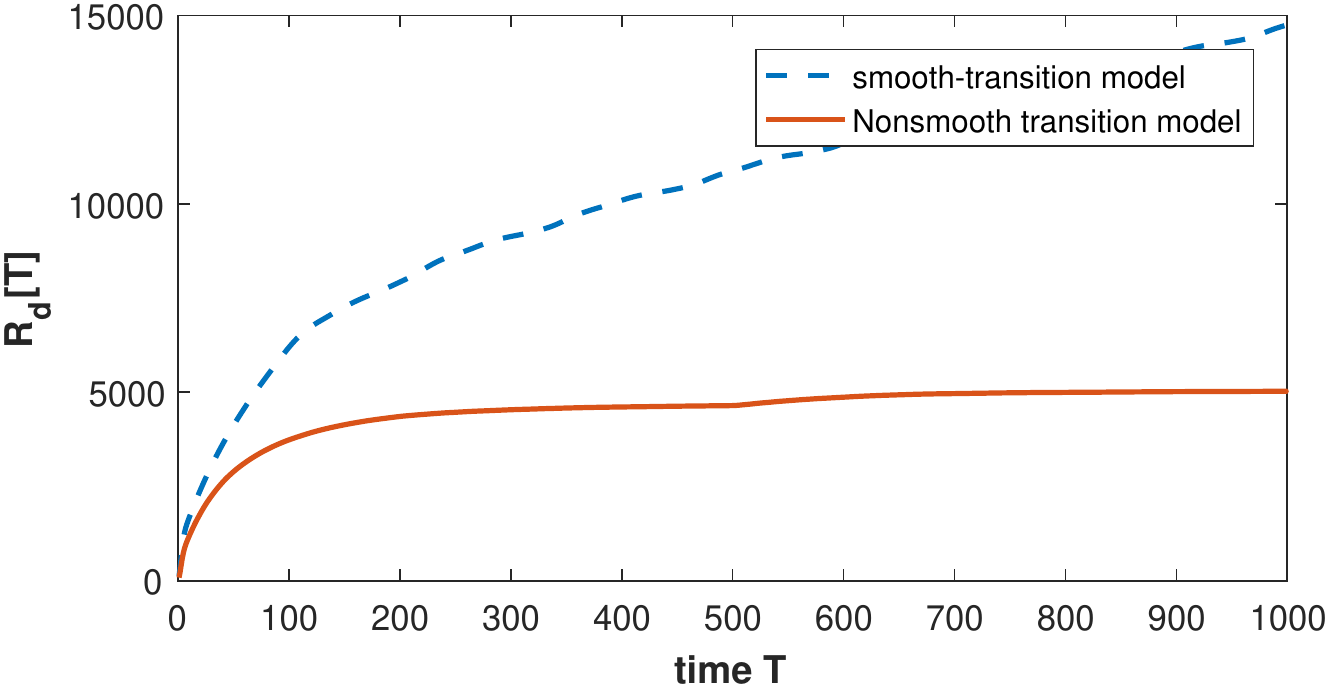}
				\caption{Dynamic regret vs time ($T$)}
				\label{fig:dregret}
	\end{subfigure}
	
		\caption{MSE versus time. Parameters: $N=10, p_e=0.15, C=5, \sigma= 0.1, \lambda=15, \gamma = 0.9, \alpha=1/L_f.$ }
\end{figure}

\section{Conclusion} \label{sec:conclusion}
An online algorithm for tracking dynamic SEM-based topologies is presented in this paper. A bound was derived on the dynamic regret (a much better metric than static regret for time-varying scenarios) of the proposed algorithm. This bound is a function of the numeric properties of the data that are easy to obatin, parameters of the algorithm, and the path length, which is a metric of how much the model parameters vary in a time interval. When the path length is sublinear in time, the dynamic regret of the algorithm becomes sublinear, meaning that the online algorithm enjoys a performance comparable to the optimal offline estimator. The tracking capabilities of the algorithm have been numerically validated for a time-varying scenario under two different assumptions on the model variations, namely a smooth-transition and an abrupt-transition model.

	\if\editmode1
	\onecolumn
	\printbibliography
	\else
	\bibliography{\bibfilenames}
	\fi
\end{document}